\documentclass[journal,onecolumn,11pt]{IEEEtran}
\usepackage{amsmath, amsthm, amssymb, amsfonts}
\usepackage{verbatim}
\usepackage{graphicx}
\usepackage{cite}
\usepackage{url}
\usepackage{enumerate}

\usepackage[bottom=0.8in]{geometry}

\newtheorem{theorem}{Theorem}
\newtheorem{lemma}{Lemma}

\newcommand{\R}{\mathbb{R}}
\newcommand{\E}{\mathbb{E}}

\newcommand{\Var}{{\rm Var}}

\newcommand{\overbar}[1]{\mkern 1.5mu\overline{\mkern-1.5mu#1\mkern-1.5mu}\mkern 1.5mu}

\def\cA{{\mathcal A}}
\def\cC{{\mathcal C}}

\def\cL{{\mathcal L}}
\def\cP{{\mathcal P}}

\def\sX{{\mathsf X}}
\def\sY{{\mathsf Y}}

\def\PP{{\mathbb P}}

\def\d{{\mathrm d}}

\def\Var{\operatorname{Var}}

\def\deq{{\triangleq}}

\def\eps{\varepsilon}
\begin{document}
%\sloppy

\title{\LARGE{Concentration of Measure Inequalities and Their Communication and Information-Theoretic Applications}}

\author{Maxim Raginsky \qquad Igal Sason
\thanks{
Maxim Raginsky is with Department of Electrical and Computer Engineering,
Coordinated Science Laboratory, University of Illinois at Urbana-Champaign,
Urbana, IL 61801, USA (e-mail: maxim@illinois.edu).}
\thanks{
I. Sason is with the Department of Electrical Engineering, Technion--Israel
Institute of Technology, Haifa 32000, Israel (e-mail: sason@ee.technion.ac.il).}}

\maketitle

\begin{abstract}
During the last two decades, concentration of measure has been a subject of various
exciting developments in convex geometry, functional
analysis, statistical physics, high-dimensional statistics, probability theory,
information theory, communications and coding theory, computer science, and learning
theory. One common theme which emerges in these fields is probabilistic stability:
complicated, nonlinear functions of a large number of independent or weakly dependent
random variables often tend to concentrate sharply around their expected values.
Information theory plays a key role in the derivation of concentration inequalities.
Indeed, both the entropy method and the approach based on transportation-cost
inequalities are two major information-theoretic paths toward proving concentration.

This brief survey is based on a recent monograph of the authors in the {\em Foundations
and Trends in Communications and Information Theory}, and a tutorial given by the authors
at ISIT~2015. It introduces information theorists to three main techniques for deriving
concentration inequalities: the martingale method, the entropy method, and the
transportation-cost inequalities.
Some applications in information theory, communications, and coding theory are used to
illustrate the main ideas.
\end{abstract}

\section{Introduction}

Concentration inequalities bound from above the
probability that a random variable $Z$ deviates from its mean,
median or some other typical value by a given amount. These
inequalities have been studied for several decades, with some
fundamental and substantial contributions during
the last two decades. Very roughly speaking, the concentration-of-measure
phenomenon can be stated in the following simple way:
``A random variable that depends in a smooth way on many independent
random variables (but not too much on any of them) is essentially
constant'' \cite{Talagrand96}. Informally, this amounts to saying that
such a random variable $Z$ concentrates around its expected value, $\E[Z]$, in such a
way that the probability of the event $\{|Z-\E[Z]| \geq t\}$,
for a given $t > 0$, decays exponentially in some power of $t$. Detailed treatments
of the concentration-of-measure phenomenon, including historical
accounts, can be found, e.g., in \cite{Boucheron_Lugosi_Massart_book,
Ledoux, Lugosi_Lecture_Notes, Massart_book, McDiarmid_tutorial, Talagrand95,
AlonS_tpm3, Raginsky_Sason_FnT_2nd}.

In recent years, concentration inequalities have been intensively studied
and used as a powerful tool in various areas. These include convex geometry,
functional analysis, statistical physics, probability theory,
statistics, information theory, communications and coding theory, learning theory,
and computer science. Several techniques have been developed so
far to prove concentration inequalities. This survey paper focuses on three
such techniques which are studied in our tutorial \cite{Raginsky_Sason_FnT_2nd}
and references therein:
\begin{itemize}
\item The martingale method (see, e.g.,
\cite{McDiarmid_tutorial,Azuma,Hoeffding}, \cite[Chapter~7]{AlonS_tpm3},
\cite{Chung_LU2006,survey2006}),
and its information-theoretic applications
(see, e.g., \cite{RiU_book} and references therein, \cite{SeldinLCTA_IT2012}).
\item The entropy method and logarithmic Sobolev inequalities
(see, e.g., \cite[Chapter~5]{Ledoux}, \cite{Lugosi_Lecture_Notes}
and references therein).
\item Transportation-cost inequalities which originated from
information theory (see, e.g., \cite[Chapter~6]{Ledoux},
\cite{Gozlan_Leonard}, \cite{Marton_ISIT2013} and references therein).
\end{itemize}
Our goal here is to give the reader a quick preview of the vast field of
concentration inequalities and their applications in information theory,
communications and coding. Therefore, we state most of the theorems and
lemmas without proofs; occasionally, we provide sketches or brief outlines.
More details can be found in our monograph \cite{Raginsky_Sason_FnT_2nd} and
the slides of our ISIT'15 tutorial.\footnote{Part 1 (The martingale method): \newline \url{http://webee.technion.ac.il/people/sason/raginsky_sason_ISIT_2015_tutorial_part_1.pdf}.
\par
Part 2 (The entropy method and transportation-cost inequalities):
\newline
\url{http://webee.technion.ac.il/people/sason/raginsky_sason_ISIT_2015_tutorial_part_2.pdf}.}

\section{The basic toolbox}

Our objective is to derive tight upper bounds on the tail probabilities
$$
\PP[Z \ge \E[Z] + t] \text{ and } \PP[Z \le \E[Z] - t], \qquad \forall t > 0
$$
where $Z = f(X_1,\ldots,X_n)$ is an arbitrary function of $n$ independent
random variables $X_1,\ldots,X_n$. To get an idea of what we can expect,
let us first recall Chebyshev's inequality:
\begin{align*}
	\PP[|Z - \E[Z]| \ge t] \le \frac{\Var[Z]}{t^2}, \qquad \forall t > 0.
\end{align*}
This inequality shows that the tail probability decays with $t$, and that the rate of decay is proportional to the variance of $Z$. Thus, the variance of $Z$ gives an idea about how tightly $Z$ concentrates around its mean. In fact, if $Z$ takes values in a bounded interval, then we can upper-bound the variance of $Z$ only in terms of the length of this interval:

\begin{lemma}\label{lm:var_bound} Let $Z$ be a random variable taking values in an interval $[a,b]$. Then
	\begin{align}\label{eq:var_bound}
	\Var[Z] \le \tfrac14 \, (b-a)^2.
\end{align}
This bound is sharp: if $Z$ only takes the two values $a$ and $b$ with equal probability, then $\Var[Z] = \tfrac14 \, (b-a)^2$.
\end{lemma}
\begin{IEEEproof} Recall that $\Var[Z] \le \E[(Z-c)^2]$ for all $c \in \R$. Letting $c = \frac{a+b}{2}$,
we obtain \eqref{eq:var_bound}. The case of equality is an easy calculation.
\end{IEEEproof}

Thus, for a bounded $Z$ in an interval $[a,b]$, Chebyshev's inequality gives $$\PP[|Z-\E[Z]| \ge t] \le \frac{(b-a)^2}{4t^2}.$$
Much stronger concentration inequalities can be derived, however, for bounded random variables. Using Markov's inequality, for every
$\lambda > 0$ we have
\begin{align*}
	\PP\left[Z - \E[Z] \ge t\right] &= \PP\left[e^{\lambda(Z-\E[Z])} \ge e^{\lambda t}\right] \\
     & \le e^{-(\lambda t-\psi(\lambda))},
\end{align*}
where $\psi(\lambda) \deq \log \E[e^{\lambda(Z - \E[Z])}]$ is the \textit{logarithmic moment-generating function} of $Z$. Optimizing over $\lambda$, we get the \textit{Chernoff bound}
\begin{align*}
	\PP\left[Z \ge \E[Z] + t \right] \le e^{-\psi^\star(t)},
\end{align*}
where $\psi^\star(t) \deq \sup_{\lambda \ge 0}\left[\lambda t - \psi(\lambda)\right]$ is the Legendre dual of $\psi$. For example, if $Z \sim N(0,\sigma^2)$ (Gaussian with mean $0$ and variance $\sigma^2$), we have $\psi(\lambda) = \lambda^2\sigma^2/2$, and $\psi^\star(t) = t^2/2\sigma^2$. With this in mind, we say that a random variable $Z$ is \textit{$\sigma^2$-subgaussian} if $\psi(\lambda) \le \lambda^2\sigma^2/2$. For a subgaussian random variable, we obtain $\psi^\star(t) \ge t^2/2\sigma^2$, which gives the tail bound
$$
\PP[Z \ge \E[Z] + t] \le e^{-t^2/2\sigma^2}, \qquad \forall t > 0.
$$
Thus, the whole affair hinges on our ability to prove that the random variable $Z$ of interest is subgaussian.

To start with, a bounded random variable is subgaussian:
\begin{lemma}[Hoeffding \cite{Hoeffding}]\label{lm:Hoeffding} For a random variable $Z$ taking values in an interval $[a,b]$, we have
	\begin{align}\label{eq:Hoeffding}
		\log \E[e^{\lambda (Z-\E[Z])}] \le \tfrac18 \, \lambda^2(b-a)^2.
	\end{align}
\end{lemma}
\begin{IEEEproof} We give a simple probabilistic proof, which has the additional benefit of highlighting the role of the tilted distribution. Let $P = \cL(Z)$,\footnote{The notation $\cL(Z)$ stands for the law, or probability distribution, of the random variable $Z$.} and introduce its \textit{exponential tilting} $P^{(t)}$: for an arbitrary sufficiently regular function $f \colon \R \to \R$,
	\begin{align*}
		\E_{P^{(t)}}[f(Z)] \deq \frac{\E_P[f(Z)e^{tZ}]}{\E_P[e^{tZ}]}.
	\end{align*}
Since $Z$ is supported on $[a,b]$ under $P$, the same holds under $P^{(t)}$ as well. Therefore, by Lemma~\ref{lm:var_bound},
$$
\Var_{P^{(t)}}[Z] \le \tfrac14 \, (b-a)^2.
$$
On the other hand,
\begin{align*}
\Var_{P^{(t)}}[Z] &= \frac{\E_P[Z^2e^{tZ}]}{\E_P[e^{tZ}]} - \left(\frac{\E_P[Ze^{tZ}]}{\E_P[e^{tZ}]}\right)^2 \\
&= \psi''(t).
\end{align*}
Therefore, $$\psi''(t) \le \tfrac14 \, (b-a)^2$$ for all $t$.
Integrating and using the fact that $$\psi(0) = \psi'(0) = 0,$$ we get \eqref{eq:Hoeffding}.
\end{IEEEproof}

Both the martingale method and the entropy method are just elaborations of these basic tools, which are applicable to an arbitrary bounded real-valued random variable. However, one should keep in mind that concentration of measure is a \textit{high-dimensional} phenomenon: we are interested in situations when $Z$ is a function of many independent random variables $X_1,\ldots,X_n$, and we can often quantify the ``sensitivity'' of $f$ to changes in each of its arguments while the others are kept fixed. This suggests that we may get a handle on the high-dimensional concentration properties of $Z$ by breaking up the problem into $n$ one-dimensional subproblems involving only one of the $X_i$'s at a time. Whenever such a divide-and-conquer approach is possible, we speak of \textit{tensorization}, by which we mean that some quantity involving the distribution of $$Z = f(X_1,\ldots,X_n)$$ (e.g., variance or relative entropy) can be related to the sum of similar quantities involving the \textit{conditional} distribution of each $X_i$ given $$\overbar{X}^i \, \deq \, (X_1,\ldots,X_{i-1},X_{i+1},\ldots,X_n).$$

\section{The martingale method}

The basic idea behind the martingale method is to start with the \textit{Doob martingale decomposition}
\begin{align}\label{eq:Doob}
Z - \E[Z] = \sum^n_{k=1} \xi_k,
\end{align}
where
\begin{align} \label{eq: xi sequence}
\xi_k \, \deq \, \E[Z|X^k] - \E[Z|X^{k-1}]
\end{align}
with
\begin{align*}
X^k \, \deq \, (X_1, \ldots, X_k)
\end{align*}
and then to exploit any information about the sensitivity of $f$ to local changes in its arguments
in order to control the sizes of the increments $\xi_k$. As a warm-up, consider the following inequality,
first obtained in a restricted setting by Efron and Stein \cite{Efron_Stein} and generalized by Steele \cite{Steele}:

\begin{lemma}[Efron--Stein--Steele]
Let $Z = f(X^n)$ where $X_1, \ldots, X_n$ are independent, then
	\begin{align}\label{eq:ESS}
		\Var[Z] \le \sum^n_{k=1} \E\left[ \Var[Z|\overbar{X}^k]\right].
	\end{align}
\end{lemma}
\begin{proof}
We exploit the fact that $\{\xi_k\}^n_{k=1}$ in \eqref{eq: xi sequence}
is a \textit{martingale difference sequence} with respect to $X^n$, i.e.,
\begin{align} \label{eq: martingale difference}
\E[\xi_k|X^{k-1}] = 0
\end{align}
for all $k \in \{1, \ldots, n\}$. Hence, since $\E[\xi_k \xi_l] = 0$ for $k \neq l$,
\begin{align}\label{eq:var_decomp}
	\Var[Z] = \sum^n_{k=1} \E[\xi^2_k].
\end{align}
The independence of $X_1,\ldots,X_n$ in \eqref{eq: xi sequence} yields
\begin{align*}
   \xi_k = \E\bigl[Z-\E[Z | \overbar{X}^k] \, | \, X^k \bigr]
\end{align*}
and, from Jensen's inequality,
\begin{align*}
   \xi_k^2 \le \E\bigl[\bigl(Z-\E[Z | \overbar{X}^k]\bigr)^2 \, | \, X^k \bigr].
\end{align*}
Due to the independence of $X_1, \ldots, X_n$, this in turn yields
\begin{align} \label{tel-aviv}
   \E[\xi_k^2] &\le \E[\bigl(Z-\E[Z | \overbar{X}^k]\bigr)^2 \bigr] \nonumber \\
               &= \E[ \Var[Z|\overbar{X}^k] ].
\end{align}
Substituting \eqref{tel-aviv} into \eqref{eq:var_decomp} yields \eqref{eq:ESS}.
\end{proof}

The Efron--Stein--Steele inequality is our first example of tensorization: it upper-bounds the
variance of $Z = f(X_1,\ldots,X_n)$ by the sum of the expected values of the conditional variances
of $Z$ given all but one of the variables. In other words, we say that $\Var[f(X_1,\ldots,X_n)]$
tensorizes. This fact has immediate useful consequences. For example, we can use any convenient
technique for upper-bounding variances to control each term on the right-hand side of \eqref{eq:var_decomp},
and thus obtain many useful variants of the Efron--Stein--Steele inequality:

\begin{enumerate}
    \item For every random variable $U$ with a finite second moment, $$\Var[U] = \tfrac12 \, \E[(U-U')^2]$$
    where $U'$ is an i.i.d.\ copy of $U$. Thus, if we let $$Z'_k = f(X_1,\ldots,X_{k-1},X'_k,X_{k+1},\ldots,X_n),$$
    where $X'_k$ is an i.i.d.\ copy of $X_k$, then $Z$ and $Z'_k$ are i.i.d. given $\overbar{X}^k$. This implies that
    $$\Var[Z|\overbar{X}^k] = \tfrac12 \, \E\Big[(Z-Z'_k)^2\Big|\overbar{X}^k\Big]$$
	for $k \in \{1,\ldots,n\}$, yielding the following variant of the Efron--Stein--Steele inequality:
	\begin{align}\label{eq:ESS_2}
	\Var[Z] \le \tfrac12 \, \sum^n_{i=1}\E[(Z-Z'_k)^2].
\end{align}
This inequality is sharp: if $Z = \sum^n_{k=1}X_k$, then $$\E[(Z-Z'_k)^2] = 2 \Var[X_k],$$ and \eqref{eq:ESS_2} holds
with equality. This shows that sums of independent random variables $X_1,\ldots,X_n$ are the least concentrated among all functions of $X^n$.

\item For every random variable $U$ with a finite second moment and for all $c \in \R$, $$\Var[U] \le \E[(U-c)^2].$$ 
Thus, by conditioning on $\overbar{X}^k$, we let $Z_k = f_k(\overbar{X}^k)$ for arbitrary functions $f_k$ ($k \in \{1, \ldots, n\}$) of $n-1$ variables to obtain
\begin{align*}
\Var[Z | \overbar{X}^k] \leq \E[(Z-Z_k)^2 | \overbar{X}^k].
\end{align*}
From \eqref{tel-aviv}, this yields another variant of the Efron--Stein--Steele inequality:
\begin{align} \label{eq:ESS_3}
	\Var[Z] \le \sum^n_{i=1} \E[(Z-Z_k)^2].
\end{align}

\item Suppose we know that, by varying just one of the arguments of $f$ while holding all others fixed, we cannot change the value of
$f$ by more than some bounded amount. More precisely, suppose that there exist finite constants $c_1,\ldots,c_n \ge 0$, such that
\begin{align}\label{eq:bdd_diff}
& \sup_x f(x_1,\ldots,x_{i-1},x,x_{i+1},\ldots,x_n) \nonumber\\
& - \inf_x f(x_1,\ldots,x_{i-1},x,x_{i+1},\ldots,x_n) \le c_i
\end{align}
for all $i$ and all $x_1,\ldots,x_{i-1},x_{i+1},\ldots,x_n$. Then, by Lemma~\ref{lm:var_bound}, 
$$
\Var[Z|\overbar{X}^k] \le \tfrac14 \, c^2_k
$$
and therefore from \eqref{eq:ESS}, \eqref{tel-aviv}
\begin{align} \label{eq:ESS_4}
\Var[Z] \le \tfrac14 \, \sum^n_{k=1}c^2_k.
\end{align}
\end{enumerate}

\subsection*{Example: Kernel Density Estimation}
As an example of Efron--Stein--Steele inequalities in action, let us look at \textit{kernel density estimation} (KDE), a nonparametric procedure for estimating an unknown pdf $\phi$ of a real-valued random variable $X$ based on observing $n$ i.i.d.\ samples $X_1,\ldots,X_n$ drawn from $\phi$ \cite[Chap.~9]{Devroye_Lugosi_DE_book}. A \textit{kernel} is a function $K \colon \R \to \R^+$ satisfying the following conditions:
\begin{enumerate}
	\item It is integrable and normalized: $\int^\infty_{-\infty} K(u) \d u = 1$.
	\item It is even: $K(u)=K(-u)$ for all $u \in \R$.
	\item $\lim_{h \downarrow 0} \frac{1}{h}K(\frac{x-u}{h}) = \delta(x-u)$, where $\delta$ is the Dirac function.
\end{enumerate}
The KDE is given by
\begin{align*}
	\phi_n(x) = \frac{1}{nh} \sum^n_{i=1} K\left(\frac{x-X_i}{h}\right),
\end{align*}
where $h > 0$ is a parameter called the \textit{bandwidth}. From the properties of $K$, for each $x \in \R$ we have
\begin{align*}
	\E[\phi_n(x)] = \frac{1}{h}\int^\infty_{-\infty} K\left(\frac{x-u}{h}\right)
	\phi(u) \d u \xrightarrow{h \downarrow 0} \phi(x).
\end{align*}
Thus, we expect the KDE $\phi_n$ to concentrate around the true pdf $\phi$; to quantify this, let us examine the $L_1$ error
$$
Z_n = f(X_1,\ldots,X_n) = \int^\infty_{-\infty} |\phi_n(x) - \phi(x)| \d x.
$$
A simple calculation shows that $f$ satisfies \eqref{eq:bdd_diff} with $$c_1 = \ldots = c_n = \frac{2}{n},$$ and
therefore \eqref{eq:ESS_4} yields $$\Var[Z_n] \le \frac1{n}.$$

\vspace*{0.2cm}
Now, to take full advantage of the martingale method, we need to combine the martingale decomposition \eqref{eq:Doob} with the Chernoff bound. To proceed, we first note that the sequence of random variables $Z_k \deq \E[Z|X^k]$, for $k=0,1,\ldots,n$, is a martingale with respect to $X_1,\ldots,X_n$, i.e., $\E[Z_{k+1}|X^k] = Z_{k}$ for each $k$. Here is one frequently used concentration result:

\begin{theorem}[Azuma--Hoeffding inequality \cite{Azuma,Hoeffding}]
\label{thm: Azuma--Hoeffding inequality}
Let $\{Z_k\}^n_{k=0}$ be a real-valued martingale sequence. Suppose that the martingale increments $\xi_k = Z_k - Z_{k-1}$, for $k=1,\ldots,n$, are almost surely bounded, i.e., $|\xi_k| \le d_k$ a.s.\ for some constants $d_1,\ldots,d_n \ge 0$. Then
\begin{align}\label{eq:Azuma}
	\PP\left[ |Z_n - Z_0| \ge t\right] \le 2 \exp\left(-\frac{t^2}{2\sum^n_{k=1}d^2_k}\right), \quad \forall t > 0.
\end{align}
\end{theorem}

\noindent The main idea behind the proof is to apply Hoeffding's lemma to each term $\xi_k$ in the Doob martingale decomposition \eqref{eq:Doob}, conditionally on $X^{k-1}$: for all $\lambda > 0$
	\begin{align*}
		\E[e^{\lambda(Z_n - Z_0)}] &= \E\left[ \prod^n_{k=1}e^{\lambda \xi_k}\right] \\
		&= \E\left[ \prod^{n-1}_{k=1} e^{\lambda \xi_k} \E[e^{\lambda \xi_n}|X^{n-1}]\right].
	\end{align*}
Since $|\xi_n| \le d_n$, we have $\ln \E[e^{\lambda \xi_n}|X^{n-1}] \le \frac{\lambda^2 d^2_n}{2}$, by Hoeffding's lemma. Continuing in this manner and peeling off the terms $\xi_k$ one by one, we can apply the Chernoff bound and obtain \eqref{eq:Azuma}.
However, the Azuma-Hoeffding inequality is not tight in general (e.g., if $t > \sum_{k=1}^n d_k$, then the probability in the left side of
\eqref{eq:Azuma} is zero, due to the boundedness of the $\xi_k$'s, whereas its bound in the right side of \eqref{eq:Azuma} is strictly positive). One way to tighten it is to make use of additional information on the conditional variances along the martingale sequence \cite{McDiarmid_bounded_differences_Martingales_1989}:
\begin{theorem}[McDiarmid]
\label{thm: McDiarmid 89}
Let $\{Z_k\}^\infty_{k=0}$ be a martingale satisfying the following two conditions for some constants $d,\sigma > 0$:
	\begin{itemize}
		\item $|\xi_k| \le d$ for all $k$.
		\item $\Var[Z_k|X^{k-1}] = \E[|\xi_k|^2|X^{k-1}] \le \sigma^2$ for all $k$.
	\end{itemize}
	Then, for every $\alpha \ge 0$,
	\begin{align*}
		\PP\left[|Z_n - Z_0| \ge \alpha n\right] \le 2\exp\left(-n d\Big(\frac{\delta+\gamma}{1+\gamma} \Big\| \frac{\gamma}{1+\gamma} \Big)\right),
	\end{align*}
	where $\gamma = \sigma^2/d^2$, $\delta = \alpha/d$, and $d(p \| q) \deq p \ln \frac{p}{q} + (1-p) \ln \frac{1-p}{1-q}$ is the binary relative entropy function.
\end{theorem}

\par
Note that, in contrast to Theorem~\ref{thm: Azuma--Hoeffding inequality}, the martingale increments $\{\xi_k\}$
in Theorem~\ref{thm: McDiarmid 89} should be bounded by a constant $d$ which is independent of $k$.

\par
A prominent application of the martingale method is a powerful
inequality due to McDiarmid \cite{McDiarmid_bounded_differences_Martingales_1989}, also known as the bounded difference inequality:

\begin{theorem}[McDiarmid's inequality]
\label{thm: McDiarmid bounded_differences 89}
If $f$ satisfies the bounded difference property \eqref{eq:bdd_diff}, and $X_1, \ldots, X_n$
are independent random variables, then for all $t > 0$
	\begin{align}\label{eq:McD}
		\PP\left[ \left|f(X^n) - \E[f(X^n)] \right|  \ge t \right] \le 2\exp\left(-\frac{2t^2}{\sum^n_{k=1}c^2_k}\right).
	\end{align}
\end{theorem}

\noindent The strategy of the proof is similar to the one used to derive the Azuma--Hoeffding inequality. In fact, we could have used the Azuma--Hoeffding inequality to bound the tail probability in \eqref{eq:McD}; however, McDiarmid's inequality provides a factor of $4$ improvement in the exponent of the bound when $f$ is a function of $n$ independent random variables.

Here is a nice information-theoretic application of McDiarmid's inequality \cite{ISIT2013_Raginsky_Sason}. Consider a discrete memoryless channel (DMC) with input alphabet $\sX$, output alphabet $\sY$, and strictly positive transition probabilities $T(y|x)$. Fix an arbitrary distribution $P_{X^n}$ of the input $n$-block $X^n$, and let $P_{Y^n}$ denote the resulting output distribution. Then, for every input $n$-block $x^n \in \sX^n$,
\begin{align}\label{eq:KL_concentration}
\PP_{Y^n|X^n=x^n}\Big[\log \frac{P_{Y^n|X^n=x^n}(Y^n)}{P_{Y^n}(Y^n)} 
\ge D(P_{Y^n|X^n=x^n}\|P_{Y^n}) + t \Big] \le \exp\left(-\frac{2t^2}{n c(T)}\right),
\end{align}
where
\begin{align}\label{eq:channel_ratio}
	c(T) \deq 2\max_{x,x' \in \sX}\max_{y \in \sY} \log \frac{T(y|x)}{T(y|x')}.
\end{align}
\begin{proof}
Let us consider the function
$$
f(y_1,\ldots,y_n) \deq \log \frac{P_{Y^n|X^n=x^n}(y^n)}{P_{Y^n}(y^n)}
$$
(recall that the input block $x^n$ is fixed). A simple calculation shows that this $f$ has bounded differences with
$$c_1 = \ldots = c_n = c(T).$$ Moreover, since the channel is memoryless, $Y_1,\ldots,Y_n$ are independent random variables under $P_{Y^n|X^n=x^n}$ (although not under $P_{Y^n}$, unless $P_{X^n}$ is a product distribution). Applying McDiarmid's inequality, we get \eqref{eq:KL_concentration}.
\end{proof}

The martingale method has also been used successfully to analyze concentration properties of random codes around their ensemble averages. The performance analysis of a particular code is usually difficult, especially for codes of large block lengths. Availability of a concentration result for the performance of capacity-approaching code ensembles under low-complexity decoding algorithms, as it is the case with low-density parity-check (LDPC) codes \cite{RiU_book}, validates the use of the density evolution technique as an analytical tool to assess the performance of individual codes from a code ensemble whose block length is sufficiently large, and to assess their asymptotic
gap to capacity. However, it should be borne in mind that the current concentration results for codes defined on graphs, which
mainly rely on the Azuma--Hoeffding inequality, are weak since in practice concentration is observed at much shorter block lengths.

Here are two illustrative examples of the use of martingale concentration inequalities in the analysis of code performance. The first result, due to Sipser and Spielman \cite{SipserS96}, is useful for assessing the performance of bit-flipping decoding algorithms for expander codes:

\begin{theorem}[Sipser and Spielman]
Let $\mathcal{G}$ be a bipartite graph that is chosen uniformly at random from
the ensemble of bipartite graphs with $n$ vertices on the left, a left degree $l$,
and a right degree $r$.
Let $\alpha \in (0,1)$ and $\delta > 0$ be fixed numbers.
Then, with probability at least $1-\exp(-\delta n)$, all sets of
$\alpha n$ vertices on the left side of $\mathcal{G}$
are connected to at least
\begin{equation*}
n \left[ \frac{l \bigl(1-(1-\alpha)^r\bigr)}{r} - \sqrt{2l \alpha \,
\bigl(h(\alpha)+\delta\bigr)} \, \right]
%\label{eq: number of neighbors}
\end{equation*}
vertices (neighbors) on the right side of $\mathcal{G}$,
where $h$ is the binary entropy function to base~$e$
(i.e., $h(x) = -x \ln(x) - (1-x) \ln(1-x), \, x \in [0,1]$).
\label{theorem: expansion}
\end{theorem}
\noindent The proof revolves around the analysis of the so-called \textit{neighbor exposure martingale} via the Azuma--Hoeffding inequality to bound the probability that the number
of neighbors deviates significantly from its mean value.

\par
Let ${\rm LDPC}(n, \lambda, \rho)$ denote an LDPC code ensemble of block length $n$, respectively, and with left and right degree distributions $\lambda$ and $\rho$ from the edge perspective (i.e., $\lambda_i$ designates the fraction of edges which are connected to a variable node of degree $i$, and $\rho_i$ designates the fraction of edges which are connected to parity-check nodes of degree $i$).

\par
The second result, due to Richardson and Urbanke \cite{RichardsonU2001}, concerns the performance of message-passing decoding algorithms for LDPC codes.

\begin{theorem}[Richardson--Urbanke]
Let $\mathcal{C}$, a code chosen uniformly at random from the
ensemble ${\rm LDPC}(n, \lambda, \rho)$, be used for transmission
over a memoryless binary-input output-symmetric (MBIOS) channel.
Assume that the decoder performs $\ell$ iterations of message-passing
decoding, and let $P_{{\rm b}}(\mathcal{C}, \ell)$ denote the
resulting bit error probability. Then, for every $\delta
> 0$, there exists some $\alpha = \alpha(\lambda,
\rho, \delta, \ell) > 0$ (independent of the block length $n$),
such that
\begin{equation*}
\PP\left[ |P_{{\rm b}}(\mathcal{C}, \ell) -
\E_{{\rm LDPC}(n,\lambda,\rho)}
[P_{{\rm b}}(\mathcal{C}, \ell)] | \geq \delta \right] \leq
e^{-\alpha n}
\end{equation*}
\end{theorem}
\noindent The proof also applies the Azuma--Hoeffding inequality
to a certain martingale sequence. Some additional references on
the use of the martingale method in the context of codes include
\cite{SipserS96,RichardsonU2001,RiU_book,LubyMSS_IT01,KavcicMM_IT2003,
Montanari05,MeassonMU08,Sason_Eshel_ISIT11}. For more details, we
refer the reader to our monograph \cite{Raginsky_Sason_FnT_2nd}.

\section{The entropy method and logarithmic Sobolev inequalities}

The entropy method, as its name suggests, relies on information-theoretic techniques to control the logarithmic moment-generating function $\psi$ directly in terms of certain relative entropies. Recall our roadmap for proving a concentration inequality for $Z = f(X)$, where $X$ is an arbitrary random variable:
\begin{itemize}
	\item Derive a tight quadratic bound on $\psi$:
	$$
	\psi(\lambda) = \log \E[e^{\lambda(Z-\E[Z])}] \le \tfrac12 \, \lambda^2 \sigma^2.
	$$
	\item Use the Chernoff bound to get
	$$
	\PP[Z \ge \E[Z] + t] \le e^{-\frac{t^2}{2\sigma^2}}, \quad \forall \, t \ge 0.
	$$
\end{itemize}
Let $P = \cL(X)$, and introduce the tilted distribution $P^{(\lambda f)}$:
$$
\d P^{(\lambda f)} = \frac{e^{\lambda f} \d P}{\E_P[e^{\lambda f}]}.
$$
The entropy method revolves around the relative entropy $D(P^{(\lambda f)} \| P)$, and has two ingredients: (1) the Herbst argument, and (2) tensorization.

We start with the Herbst argument (the name refers to an unpublished note by I.~Herbst that proposed the use of such an argument in the context of mathematical physics of quantum fields). Let us examine the relative entropy:
\begin{align*}
	D(P^{(\lambda f)} \| P) &= \int \d P^{(\lambda f)} \log \frac{\d P^{(\lambda f)}}{\d P} \\
	&= \E^{(\lambda f)}\left[ \lambda f(X) - \psi(\lambda)\right] \\
	&= \lambda \psi'(\lambda) - \psi(\lambda),
\end{align*}
where $\E^{(\lambda f)}[\cdot]$ denotes expectation with respect to the tilted distribution $P^{(\lambda f)}$.
Now, with a bit of foresight, we rewrite the last expression as
$$
\lambda \psi'(\lambda) - \psi(\lambda) = \lambda^2 \frac{\d}{\d\lambda}\left(\frac{\psi(\lambda)}{\lambda}\right).
$$
Thus, we end up with the identity
\begin{align*}
    D(P^{(\lambda f)} \| P) = \lambda^2 \, \frac{\mathrm{d}}{\mathrm{d} \lambda} \left( \frac{\psi(\lambda)}{\lambda} \right).
\end{align*}
Integrating and using the fact that $\lim_{\lambda \to 0} \frac{\psi(\lambda)}{\lambda} = 0$ (which can be proved using l'Hopital's rule), we get
\begin{align}\label{eq:Herbst_identity}
	\psi(\lambda) = \lambda \int^\lambda_0 \frac{D(P^{(tf)} \| P)}{t^2} \d t.
\end{align}
Appealing to the Chernoff bound, we end up with the following:
\begin{lemma}[The Herbst argument] Suppose that $Z = f(X)$ is such that
	\begin{align}\label{eq:Herbst_bound}
		D(P^{(\lambda f)} \| P) \le \tfrac12 \, \lambda^2 \sigma^2, \qquad \forall \, \lambda \ge 0.
	\end{align}
	Then $Z$ is $\sigma^2$-subgaussian, and therefore
	\begin{align}
		\PP\left[f(X) \ge \E[f(X)] + t \right] \le e^{-\frac{t^2}{2\sigma^2}}, \qquad \forall \, t \ge 0.
	\end{align}
\end{lemma}
\noindent In fact, it can be shown that the reverse implication holds as well, but with some loss in the constants \cite{Ramon_lectures}: if $Z = f(X)$ is $\sigma^2/4$-subgaussian, then
	\begin{align*}
		D(P^{(\lambda f)} \| P) \le \tfrac12 \, \lambda^2\sigma^2, \qquad \lambda \ge 0.
	\end{align*}
In other words, subgaussianity of $Z = f(X)$ is equivalent to $D(P^{(\lambda f)} \| P) = O(\lambda^2)$. It seems, therefore, that we have not really accomplished anything, apart from arriving at an equivalent characterization of subgaussianity. However, the relative entropy has one crucial property: it tensorizes. Recall that we are interested in the high-dimensional setting, where $X = (X_1,\ldots,X_n)$ is a tuple of $n$ independent random variables. Thus, $P = \cL(X)$ is a product distribution: $P_{X} = P_{X_1} \otimes \ldots \otimes P_{X_n}$. Using this fact together with the chain rule for relative entropy, we arrive at the following:
\begin{lemma}[Tensorization of the relative entropy] Let $P$ and $Q$ be two probability distributions of a random $n$-tuple $X = (X_1,\ldots,X_n)$, such that the coordinates of $X$ are independent under $P$. Then
	\begin{align}\label{eq:KL_tensorization}
		D(Q \| P) \le \sum^n_{i=1} D(Q_{X_i|\overbar{X}^i} \| P_{X_i} | Q_{\overbar{X}^i}).
	\end{align}
\end{lemma}
\noindent The quantity on the right-hand side of \eqref{eq:KL_tensorization} is the \textit{erasure divergence} between $Q$ and $P$ \cite{Verdu_Weissman_erasures}. We now particularize this general bound to our problem, where $Q$ is given by the tilted distribution $P^{(\lambda f)}$. In that case, using Bayes' rule and the fact that the $X_i$'s are independent, we can express the conditional distributions $P^{(\lambda f)}_{X_i|\overbar{X}^i}$ as follows: for each $\overbar{x}^i$,
\begin{align*}
	\d P^{(\lambda f)}_{X_i|\overbar{X}^i=\overbar{x}^i} = \frac{e^{\lambda f(x_1,\ldots,x_{i-1},\cdot,x_{i+1},\ldots,x_n)}}{\E\left[e^{\lambda f(x_1,\ldots,x_{i-1},X_i,x_{i+1},\ldots,x_n)}\right]} \; \d P_{X_i}.
\end{align*}
This looks formidable; nevertheless, it reveals that the conditional distribution $P^{(\lambda f)}_{X_i|\overbar{X}^i=\overbar{x}^i}$ is the exponential tilting of the marginal distribution $P_{X_i}$ with respect to the random variable $f_i(X_i) = f(x_1,\ldots,x_{i-1},X_i,x_{i+1},\ldots,x_n)$, which depends only on $X_i$ because $\overbar{x}^i$ is fixed. Thus, we arrive at the following bound:
\begin{align*}
	D(P^{(\lambda f)} \| P) \le \sum^n_{i=1} \widetilde{\E}\left[D(P^{(\lambda f_i)}_{X_i} \| P_{X_i})\right],
\end{align*}
where the expectation on the right-hand side is with respect to the tilted distribution.

We can now distill the entropy method into a series of steps:
\begin{enumerate}
	\item We wish to derive a subgaussian tail bound
	$$
	\PP\left[f(X^n) \ge \E[f(X^n)] + t \right] \le e^{-\frac{t^2}{2\sigma^2}}, \qquad t \ge 0,
	$$
	where $X_1,\ldots,X_n$ are independent random variables.
	\item Suppose that we can prove that there exist constants $c_1,\ldots,c_n \ge 0$, such that
	\begin{align}\label{eq:divide_and_conquer}
	D(P^{(\lambda f_i)}_{X_i} \| P_{X_i}) \le \tfrac12 \, \lambda^2 c^2_i, \qquad \forall \, i \in \{1, \ldots, n\}.
    \end{align}
	\item Then, by the tensorization lemma,
	$$
	D(P^{(\lambda f)} \| P) \le \tfrac12 \, \lambda^2 \sum^n_{i=1}c^2_i,
	$$
	and therefore, by the Herbst argument, $Z = f(X^n)$ is $\sigma^2$-subgaussian with $\sigma^2 = \sum^n_{i=1}c^2_i$.
\end{enumerate}
The main benefit of passing to the relative-entropy characterization of subgaussianity is that now, via tensorization, we have broken up a difficult $n$-dimensional problem into $n$ presumably easier $1$-dimensional problems, each of which boils down to analyzing the behavior of the function $f_i(X_i) \equiv f(x_1,\ldots,x_{i-1},X_i,x_{i+1},\ldots,x_n)$, where only the $i$th input coordinate is random, and the remaining ones are fixed at some arbitrary values.

Of course, the problem now reduces to showing that \eqref{eq:divide_and_conquer} holds. One route, which often yields tight constants, is via so-called \textit{logarithmic Sobolev inequalities}. In a nutshell, a logarithmic Sobolev inequality (or LSI, for short) ties together a probability distribution $P$, some function class $\cA$ that contains the function $f$ of interest, and an ``energy" functional $E : \cA \to \R$ with the property
$$
E(\alpha f) = \alpha E(f), \qquad \forall \alpha \ge 0, f \in \cA.
$$
With these ingredients in place, a log-Sobolev inequality takes the form
$$
D(P^{(f)} \| P) \le \tfrac12 \, c \, E^2(f), \qquad \forall f \in \cA.
$$
Now suppose that $E(f) \le L$. Then we readily get the bound
$$
D(P^{(\lambda f)} \| P) \le \tfrac12 \, c \, E^2(\lambda f ) = \tfrac12 \, \lambda^2 c \, E^2(f) \le \tfrac12 \, \lambda^2 cL^2,
$$
so $f(X)$, $X \sim P$, is $\sigma^2$-subgaussian with $\sigma^2 = cL^2$.

There is a vast literature on log-Sobolev inequalities, and an interested reader may consult our monograph for more details and additional references. Here we will give the two classic examples: the Bernoulli LSI and the Gaussian LSI, due to Gross \cite{Gross}.
\begin{theorem}[Bernoulli LSI] Let $X_1,\ldots,X_n$ be i.i.d.\ ${\rm Bern}(1/2)$ random variables. Then, for every function $f \colon \{0,1\}^n \to \R$, we have
	\begin{align}
		D(P^{(f)} \| P) \le \frac{1}{8} \frac{\E\left[|Df(X^n)|^2 e^{f(X^n)}\right]}{\E[e^{f(X^n)}]},
	\end{align}
	where $P = {\rm Bern}(1/2)^{\otimes n}$,
	$$
	Df(x^n) \deq \sqrt{\sum^n_{i=1} \left|f(x^n) - f(x^n \oplus e_i)\right|^2},
	$$
	and $x^n \oplus e_i$ is the XOR of $x^n$ with the bit string of all zeros, except for the $i$th bit. In other words, $x^n \oplus e_i$ is $x^n$ with the $i$th bit flipped.
\end{theorem}
\noindent The proof, which we omit, is to first establish the $n=1$ case via a straightforward if tedious exercise in calculus, and then to extend to an arbitrary $n$ by tensorization. Note that the mapping $f \mapsto Df$ has the desired scaling property: $D(\alpha f) = \alpha D(f)$ for all $\alpha \ge 0$.

\begin{theorem}[Gaussian LSI] Let $X_1,\ldots,X_n$ be i.i.d. $N(0,1)$ random variables. Then, for an arbitrary smooth function
$f \colon \R^n \to \R$,
	\begin{align}
		D(P^{(f)}\|P) \le \frac{1}{2} \frac{\E\left[\| \nabla f(X^n) \|^2_2 e^{f(X^n)}\right]}{\E[e^{f(X^n)}]}.
	\end{align}
\end{theorem}
\noindent Note that the mapping $f \mapsto \| \nabla f \|_2$ has the scaling property: $\| \nabla (\alpha f) \|_2 = \alpha \| \nabla f \|_2$ for all $\alpha \ge 0$. By now, there are at least fifteen different ways in the literature for proving the Gaussian LSI. The original proof by Gross was to apply the Bernoulli LSI to the function
$$
f\left(\frac{X_1 + \ldots + X_n - n/2}{\sqrt{n/4}}\right), \qquad X_i \stackrel{{\rm i.i.d.}}{\sim} {\rm Bern}(1/2),
$$
and then pass to the Gaussian limit by appealing to the Central Limit Theorem.

The Gaussian LSI can be used to give a short proof of the following concentration inequality for Lipschitz functions of Gaussians, which was originally obtained by Tsirelson, Ibragimov, and Sudakov \cite{Tsirelson} using different methods:
\begin{theorem}[Tsirelson--Ibragimov--Sudakov] Let $X_1,\ldots,X_n$ be i.i.d.\ $N(0,1)$ random variables, and let $f \colon \R^n \to \R$ be a function which is $L$-Lipschitz:
	$$
	\left|f(x^n)-f(y^n)\right| \le L \left\| x^n - y^n \right\|_2.
	$$
Then, $f(X^n)$ is $L^2$-subgaussian, which yields
	\begin{align}\label{eq:Gauss_Lip}
		\PP\left[f(X^n) \ge \E[f(X^n)] + t \right] \le e^{-\frac{t^2}{2L^2}}
	\end{align}
for all $t > 0$.
\end{theorem}
\begin{proof} By a standard approximation argument, we may assume that $f$ is differentiable. Since it is also $L$-Lipschitz, $\| \nabla f \|^2_2 \le L^2$ everywhere. Substituting this bound into the Gaussian LSI for $\lambda f$, we obtain
	$$
	D(P^{(\lambda f)} \| f) \le \tfrac12 \, \lambda^2 L^2.
	$$
By the Herbst argument, $Z = f(X^n)$, $X^n \sim N(0,I_n)$, is $L^2$-subgaussian, and we are done.
\end{proof}

\noindent This result is remarkable in two ways: It only assumes Lipschitz continuity of $f$, and gives \textit{dimension-free} concentration (i.e., the exponent in \eqref{eq:Gauss_Lip} does not depend on $n$).

Deriving log-Sobolev inequalities, especially with tight constants, is a subtle art. A commonly used method is to realize $P$ as an invariant distribution of some continuous-time reversible Markov process and to extract a suitable energy functional $E$ from the structure of the infinitesimal generator of the process. In many cases, however, it is possible to derive a log-Sobolev inequality via tensorization and a nice and simple variance-based representation of the relative entropy due to A.~Maurer~\cite{Maurer_thermo}:

\begin{theorem}[Maurer]
\label{thm:Maurer_thermo}
Let $X$ be a random variable with law $P$. Then, for every real-valued function $f$ and all $\lambda \ge 0$
$$
D(P^{(\lambda f)} \| P) = \int^\lambda_0 \int^\lambda_t \Var^{(sf)}[f(X)] \d s\, \d t,
$$
where $\Var^{(sf)}[f(X)]$ is the variance of $f(X)$ under the tilted distribution $P^{(s f)}$.
\end{theorem}
\begin{proof} As before, let $\psi(\lambda) = \log \E[e^{\lambda(f(X))-\E[f(X)]}]$ be the logarithmic moment-generating function of $f(X)$. Then
	\begin{align*}
		D(P^{(\lambda f)} \| P) &= \lambda \psi'(\lambda) - \psi(\lambda) \\
		&= \int^\lambda_0 \left[\psi'(\lambda)-\psi'(t)\right] \d t \\
		&= \int^\lambda_0 \int^\lambda_t \psi''(s) \d s\, \d t,
	\end{align*}
	where we have used the fact that $\psi(0) = \psi'(0) = 0$ and the fundamental theorem of calculus. Recalling that $\psi''(s) = \Var^{(sf)}[f(X)]$, we are done.
\end{proof}
The following result is a direct consequence of Theorem~\ref{thm:Maurer_thermo}:

\begin{theorem}\label{thm:Maurer_method} Let $\cA$ be a class of functions of $X$, and suppose that there is a mapping $\Gamma : \cA \to \R$, such that:
\begin{enumerate}
	\item For all $f \in \cA$ and $\alpha \ge 0$, $\Gamma(\alpha f) = \alpha \Gamma(f)$.
	\item There exists a constant $c > 0$, such that
	$$
	\Var^{(\lambda f)}[f(X)] \le c |\Gamma(f)|^2, \qquad \forall f \in \cA,\, \lambda \ge 0.
	$$
\end{enumerate}
Then
\begin{align*}
	D(P^{(\lambda f)} \| P) \le \tfrac12 \, \lambda^2 c \, |\Gamma(f)|^2, \qquad \forall f \in \cA, \lambda \ge 0.
\end{align*}
\end{theorem}
\noindent To illustrate Maurer's method, let's use it to derive the Bernoulli LSI. It suffices to prove the $n=1$ case, and then to scale up to an arbitrary $n$ by tensorization. Thus, let $P = {\rm Bern}(1/2)$, and for every function $f \colon \{0,1\} \to \R$ define $\Gamma (f) \deq |f(0)-f(1)|$. By Lemma~\ref{lm:var_bound},
\begin{align*}
	\Var^{(\lambda f)}[f(X)] \le \tfrac14 \, |f(0)-f(1)|^2 = \tfrac14 \, |\Gamma(f)|^2.
\end{align*}
Thus, the conditions of Theorem~\ref{thm:Maurer_method} are satisfied with $c = 1/4$, and we get precisely the Bernoulli LSI. One can also use Maurer's method to prove McDiarmid's inequality (see Theorem~\ref{thm: McDiarmid bounded_differences 89}).

\section{Transportation-cost inequalities}

At this point, we notice a common theme running through the above examples of concentration:
\begin{itemize}
	\item Let $f \colon \R^n \to \R$ be $1$-Lipschitz with respect to the Euclidean norm $\| \cdot \|_2$, and let $X_1,\ldots,X_n$ be i.i.d.\ $N(0,1)$ random variables. Then, for every $t \ge 0$,
	$$
	\PP[f(X^n) \ge \E[f(X^n)] + t] \le e^{-t^2/2}.
	$$
	\item Let $\sX$ be an arbitrary space, and consider a function $f \colon \sX^n \to \R$, which is $1$-Lipschitz with respect to the weighted Hamming metric
	$$
	d_{{\bf c}}(x^n,y^n) \deq \sum^n_{i=1} c_i {\bf 1}_{\{x_i \neq y_i\}},
	$$
	where $c_1,\ldots,c_n \ge 0$ are some fixed constants. It is easy to see that such a Lipschitz property is equivalent to the bounded difference property \eqref{eq:bdd_diff}, and in that case McDiarmid's inequality tells us that
	$$
	\PP[f(X^n) \ge \E[f(X^n)] + t] \le e^{-2t^2/\sum^n_{i=1}c^2_i}
	$$
	for every tuple $X_1,\ldots,X_n$ of independent $\sX$-valued random variables.
\end{itemize}
Thus, metric spaces and Lipschitz functions seem to be a natural setting to study concentration. To make this statement more precise, let $(\sX,d)$ be a metric space. We say that a function $f \colon \sX \to \R$ is $L$-Lipschitz (with respect to $d$) if
$$
|f(x)-f(y)| \le L d(x,y), \qquad \forall x,y \in \sX.
$$
Denoting by ${\rm Lip}_L(\sX,d)$ the class of all $L$-Lipschitz functions, we can pose the following question: What conditions does a probability distribution $P$ on $\sX$ have to satisfy, so that $f(X)$ with $X \sim P$ is $\sigma^2$-subgaussian for every $f \in {\rm Lip}_1(\sX,d)$?

Through the pioneering work of Katalin Marton \cite{Marton_blowup,Marton_contracting_MCs,Marton_dbar,Marton_Euclidean,Marton_Euclidean_correction,Marton_ISIT2013}, the answer to the above question has deep links to information theory via the notion of so-called \textit{transportation-cost inequalities} \cite{Villani_TOT}. In order to introduce them, we first need some definitions. A \textit{coupling} of two probability distributions $P$ and $Q$ on $\sX$ is a probability distribution $\pi$ on the Cartesian product $\sX \times \sX$, such that for $(X,Y) \sim \pi$ we have $X \sim P$ and $Y \sim Q$. Let $\Pi(P,Q)$ denote the set of all couplings of $P$ and $Q$. For $p \ge 1$, the \textit{$L^p$ Wasserstein distance} between $P$ and $Q$ is defined as
	\begin{align*}
		W_p(P,Q) \deq \inf_{\pi \in \Pi(P,Q)} \left(\E_\pi[d^p(X,Y)]\right)^{1/p}.
	\end{align*}
The name ``transportation cost'' comes from the following interpretation: Let $P$ (resp., $Q$) represent the initial (resp., desired) distribution of some matter (say, sand) in space, such that the total mass in both cases is normalized to one. Thus, both $P$ and $Q$ correspond to sand piles of some given shapes. The objective is to rearrange the initial sand pile with shape $P$ into one with shape $Q$ with minimum cost, where the cost of transporting a grain of sand from location $x$ to location $y$ is given by $d^p(x,y)$. If we allow randomized transportation policies, i.e., those that associate with each location $x$ in the initial sand pile a conditional probability distribution $\pi(\d y|x)$ for its destination in the final sand pile, then the minimum transportation cost is given by $W_p(P,Q)$. We say that $P$ satisfies an $L^p$ \textit{transportation-cost inequality} with constant $c$, or $T_p(c)$ for short, if
\begin{align*}
	W_p(P,Q) \le \sqrt{2 c D(Q \| P)}, \qquad \forall Q.
\end{align*}
The well-known Pinsker's inequality is, in fact, a transportation-cost inequality: If we take $\sX$ to be an arbitrary space and equip it with the metric $d(x,y) = {\bf 1}_{\{x \neq y\}}$, then the $L^1$ Wasserstein distance $W_1(P,Q)$ is simply the total variation distance
$$
\| P - Q \|_{\rm TV} = \sup_{A} |P(A)-Q(A)|,
$$
and Pinsker's inequality $$\| P - Q \|_{\rm TV} \le \sqrt{\tfrac12 \, D(Q \| P)}$$ (in nats) is then a $T_1(\tfrac14)$ inequality, which is satisfied by all probability measures $P, Q$ where $Q \ll P$ (i.e., $Q$ is absolutely continuous with respect to $P$). Various distribution-dependent refinements of Pinsker's inequality where the constant is optimized for a fixed $P$ while varying only $Q$ \cite{Ordentlich_Weinberger_Pinsker,Berend_Harremoes_Kontorovich} can be interpreted in the same vein as well. Another well-known transportation-cost (TC) inequality is due to Talagrand \cite{Talagrand_Gaussian_T2}: Let $\sX$ be the Euclidean space $\R^n$, equipped with the Euclidean metric $d(x,y) = \| x-y \|_2$. Then $P = N(0,I_n)$ satisfies the $T_2(1)$ inequality: $W_2(P,Q) \le \sqrt{2 D(Q \| P)}$. The remarkable thing here is that the constant is independent of the dimension $n$.

With these preliminaries out of the way, we can now state the theorem, due to Bobkov and G\"otze \cite{Bobkov_Gotze_expint}, which provides an answer to the question posed above:
\begin{theorem}[Bobkov--G\"otze] Let $X$ be a random variable taking values in a metric space $(\sX,d)$ according to a probability distribution $P$. Then, the following are equivalent:
	\begin{enumerate}
		\item $f(X)$ is $\sigma^2$-subgaussian for every $f \in {\rm Lip}_1(\sX,d)$.
		\item $P$ satisfies $T_1(\sigma^2)$, i.e.,
		$$
		W_1(P,Q) \le \sqrt{2\sigma^2 D(Q \| P)}
		$$
		for all $Q$.
	\end{enumerate}
\end{theorem}

At this point, one may wonder what we have gained -- verifying that a given $P$ satisfies a TC inequality, let alone determining tight constants, is a formidable challenge. However, once again, tensorization comes to the rescue. Marton's insight was that TC inequalities tensorize \cite{Villani_TOT}:

\begin{theorem}Let $(\sX_i,P_i,d_i)$, $1 \le i \le n$, be probability metric spaces. If for some $1 \le p \le 2$ each $P_i$ satisfies $T_p(c)$ on $(\sX_i,d_i)$, then the product measure $P = P_1 \otimes \ldots \otimes P_n$ on $\sX = \sX_1 \times \ldots \times \sX_n$ satisfies $T_p(cn^{2/p-1})$ w.r.t.\ the metric
$$
d_p(x^n,y^n) \deq \left( \sum^n_{i=1} d^p_i(x_i,y_i)\right)^{1/p}.
$$
\end{theorem}
\noindent In particular, if each $P_i$ satisfies $T_1(c)$, then $P = P_1 \otimes \ldots \otimes P_n$ satisfies $T_1(cn)$ with respect to the metric $\sum_i d_i$. Note that the constant deteriorates with $n$. On the other hand, if each $P_i$ satisfies $T_2(c)$, then $P$ satisfies $T_2(c)$ with respect to $\sqrt{\sum_i d^2_i}$. Note that the latter constant is independent of $n$.

To give a simple illustration of all these concepts, let us outline yet another proof of McDiarmid's inequality. Consider a product probability space $(\sX_1 \times \ldots \times \sX_n, P_1 \otimes \ldots \otimes P_n)$. For a fixed choice of constants $c_1,\ldots,c_n \ge 0$, equip $\sX_i$ with the metric $d_i(x_i,y_i) = c_i {\bf 1}_{\{x_i \neq y_i\}}$. Then, by rescaling Pinsker's inequality, we see that $P_i$ satisfies a $T_1(c^2_i/4)$ inequality with respect to the metric $d_i$:
\begin{align}
	W_{1,d_i}(P_i,Q_i) \le \sqrt{\tfrac12 \, c^2_i D(Q_i \| P_i)}, \qquad \forall Q_i.
\end{align}
By the tensorization theorem for TC inequalities, the product distribution $P$ satisfies a $T_1(c)$ inequality with $c = (1/4)\sum^n_{i=1}c^2_i$ with respect to the weighted Hamming metric $d_{{\bf c}}$. By the Bobkov--G\"otze theorem, this is equivalent to the subgaussianity of all $f(X_1,\ldots,X_n)$ with $f \in {\rm Lip}_1(\sX,d)$ and mutually independent $X_i \in \sX_i$, $1 \le i \le n$. But this is precisely McDiarmid's inequality.

\section{Some applications in information theory}

We end this survey by briefly describing some information-theoretic applications of concentration inequalities.

\subsection{The Blowing-up Lemma and Information-Theoretic Consequences}
The first explicit appeal to the concentration phenomenon in information
theory dates back to the 1970s work by Ahlswede and collaborators, who used
the so-called \textit{blowing-up lemma} for deriving strong converses for a variety of communications
and coding problems.

Consider a product space $\sY^n$ equipped with the Hamming metric $d(y^n,z^n) = \sum^n_{i=1} {\bf 1}_{\{y_i \neq z_i\}}$. For $r \in \{0,1,\ldots,n\}$, define the \textit{$r$-blowup} of a  set $A \subseteq \sY^n$  as
	$$
	[A]_r \deq \left\{ z^n \in \sY^n : \min_{y^n \in A} d(z^n,y^n) \le r \right\}
	$$
The following result, in a different (asymptotic) form was first proved by Ahlswede, G\'acs, and K\"orner \cite{Ahlswede_Gacs_Korner}; a simple proof, which we sketch below, was given by Marton \cite{Marton_blowup}:

\begin{lemma}[Blowing-up] Let $Y_1,\ldots,Y_n$ be independent random variables taking values in $\sY$. Then for every set $A \subseteq \sY^n$ with $P_{Y^n}(A) > 0$
$$
P_{Y^n}\left\{ \left[A\right]_{r}\right\} \ge 1 - \exp\left[-\frac{2}{n}\left(r-\sqrt{\frac{n}{2}\log \frac{1}{P_{Y^n}(A)}}\right)^2_+\right],
$$
where $(u)_+ \deq \max\{0,u\}$.
\end{lemma}
\begin{proof} We sketch the proof in order to highlight the role of TC inequalities. For each $i \in \{1, \ldots, n\}$, let $P_i = \cL(Y_i)$. By tensorization, the product distribution $P = P_{Y^n}$ satisfies the TC inequality
	\begin{align}\label{eq:blowup_TC}
		W_1(P,Q) \le \sqrt{\frac{n}{2}D(Q \| P)}, \qquad \forall Q,
	\end{align}
	where
	\begin{align*}
		W_1(P,Q) = \inf_{\pi \in \Pi(P,Q)} \E_\pi\left[\sum^n_{i=1}{\bf 1}_{\{X_i \neq Y_i\}}\right].
	\end{align*}
	Now, for an arbitrary $B \subseteq \sY^n$ with $P(B) > 0$, consider the conditional distribution $P_B(\cdot) \deq \frac{P(\cdot \cap B)}{P(B)}$. Then $D(P_B \| P) = \log \frac{1}{P(B)}$, and in that case using \eqref{eq:blowup_TC} with $Q = P_B$, we get
	\begin{align}\label{eq:blowup_TC_2}
		W_1(P,P_B) \le \sqrt{\frac{n}{2}\log \frac{1}{P(B)}}.
	\end{align}
Applying \eqref{eq:blowup_TC} to $B = A$ and $B = [A]^c_r$, we get
\begin{align*}
	W_1(P,P_A) &\le \sqrt{\frac{n}{2}\log \frac{1}{P(A)}},\\
	W_1(P,P_{[A]^c_r}) &\le \sqrt{\frac{n}{2}\log \frac{1}{1-P([A]_r)}}.
\end{align*}
Adding up these two inequalities, we obtain
\begin{align*}
&	\sqrt{\frac{n}{2}\log \frac{1}{P(A)}} + \sqrt{\frac{n}{2}\log \frac{1}{1-P([A]_r)}} \nonumber\\
& \qquad \ge W_1(P_A,P) + W_1(P_{[A]^c_r},P) \\
& \qquad \ge W_1(P_A,P_{[A]^c_r}) \\
& \qquad \ge \min_{x^n \in A, y^n \in [A]^c_r} d(x^n,y^n) \\
& \qquad \ge r,
\end{align*}
where the first step holds due to \eqref{eq:blowup_TC_2}, the second step is verified by the triangle inequality, and the remaining steps follow from definitions. Rearranging, we obtain the lemma.
\end{proof}
Informally, the lemma states that every set in a product space can be ``blown up'' to engulf most of the probability mass. Using this fact, one can prove strong converses for channel coding in single-terminal and multiterminal settings. Here is the simplest consequence of the blowing-up lemma in the context of channel codes: Consider a DMC with input alphabet $\sX$, output alphabet $\sY$, and transition probabilities $T(y|x)$, $x \in \sX, y \in \sY$. An $(n,M,\eps)$-code for $T$ consists of an encoder $f \colon \{1,\ldots,M\} \to \sX^n$ and a decoder $g : \sY^n \to \{1,\ldots,M\}$, such that
\begin{align*}
	\max_{1 \le j \le M} \PP[g(Y^n) \neq j \, | \, f(X^n) = j] \le \eps.
\end{align*}
\begin{lemma} Let $u_j = f(j)$, $1 \le j \le M$, denote the $M$ codewords of the code, and let $D_j \deq g^{-1}(j)$ be the corresponding decoding regions in $\sY^n$.  There exists some $\delta_n > 0$, such that
	$$
	T^n\Big([D_j]_{n\delta_n} \, \big| \, X^n = u_j \Big) \ge 1 - \frac{1}{n}, \qquad \forall \, j \in \{1,\ldots,M\}.
	$$
\end{lemma}
\noindent Informally, this corollary of the blowing-up lemma says that ``any bad code contains a good subcode." Using this result, Ahlswede and Dueck \cite{Ahlswede_Dueck} established a strong converse for channel coding as follows: Consider an $(n,M,\eps)$-code $\cC = \{(u_j,D_j)\}^M_{j=1}$. Each decoding set $D_j$ can be ``blown up'' to a set $\tilde{D}_j \subseteq \sY^n$ with
$$
T^n(\tilde{D}_j | u_j) \ge 1 - \frac{1}{n}.
$$
The object $\tilde{\cC} = \{(u_j,\tilde{D}_j)\}^M_{j=1}$ is not a code (since the sets $\tilde{D}_j$ are no longer disjoint), but a random coding argument can be used to extract an $(n,M',\eps')$ ``subcode'' with $M'$ slightly smaller than $M$ and $\eps' < \eps$. Then one can apply the usual (weak) converse to the subcode. Similar ideas have found use in multiterminal settings, starting with the work of Ahlswede--G\'acs--K\"orner \cite{Ahlswede_Gacs_Korner}.

\subsection{Empirical distribution of good channel codes with non-vanishing error probability}
Another recent application of concentration inequalities to information theory has to do with characterizing stochastic behavior of output sequences of good channel codes. On a conceptual level, the random coding argument originally used by Shannon (and many times since) to show the existence of good channel codes suggests that the input/output sequence of such a code should resemble, as much as possible, a typical realization of a sequence of i.i.d.\ random variables sampled from a capacity-achieving input/output distribution. For capacity-achieving sequences of codes with asymptotically vanishing probability of error, this intuition has been analyzed rigorously by Shamai and Verd\'u \cite{Shamai_Verdu_empirical}, who have proved the following remarkable statement \cite[Theorem~2]{Shamai_Verdu_empirical}: given a DMC $T$, any capacity-achieving sequence of channel codes with asymptotically vanishing probability of error (maximal or average) has the property that
\begin{align}\label{eq:convergence_to_caod}
	\lim_{n \to \infty} \frac{1}{n}D(P_{Y^n} \| P^*_{Y^n}) = 0,
\end{align}
where, for each $n$, $P_{Y^n}$ denotes the output distribution on $\sY^n$ induced by the code (assuming that the
messages are equiprobable), while $P^*_{Y^n}$ is the product of $n$ copies of the single-letter capacity-achieving output distribution. In a recent paper \cite{Polyanskiy_Verdu_good_codes}, Polyanskiy and Verd\'u extended the results of \cite{Shamai_Verdu_empirical} for codes with {\em nonvanishing} probability of error.

To keep things simple, we will only focus on channels with finite input and output alphabets. Thus, let $\sX$ and $\sY$ be finite sets, and consider a DMC $T$ with capacity $C$. Let $P^*_X \in \cP(\sX)$ be a capacity-achieving input distribution (which may be nonunique). It can be shown \cite{Topsoe_capacity} that the corresponding output distribution $P^*_Y \in \cP(\sY)$ is unique. Consider any $(n,M)$-code $\cC = (f,g)$, let $P^{(\cC)}_{X^n}$ denote the distribution of $X^n = f(J)$, where $J$ is uniformly distributed in $\{1,\ldots,M\}$, and let $P^{(\cC)}_{Y^n}$ denote the corresponding output distribution. The central result of \cite{Polyanskiy_Verdu_good_codes} is that the output distribution $P^{(\cC)}_{Y^n}$ of any $(n,M,\eps)$-code satisfies
\begin{align}\label{eq:PV_output_bound}
	D\big(P^{(\cC)}_{Y^n} \big\| P^*_{Y^n}\big) \le nC - \log M + o(n);
\end{align}
moreover, the $o(n)$ term was refined in \cite[Theorem~5]{Polyanskiy_Verdu_good_codes} to $O(\sqrt{n})$ for
any DMC, except those that have zeroes in their transition matrix. Using McDiarmid's inequality, this result
is sharpened as follows \cite{ISIT2013_Raginsky_Sason}:

\begin{theorem} Consider a DMC $T$ with positive transition probabilities. Then any $(n,M,\eps)$-code $\cC$ for $T$, with $\eps \in (0,1/2)$, satisfies
\begin{align*}
	D\Big(P^{(\cC)}_{Y^n} \Big\| P^*_{Y^n} \Big) \le nC - \log M 
	+ \log \frac{1}{\eps} + c(T)\sqrt{\frac{n}{2}\log \frac{1}{1-2\eps}},
\end{align*}
where $c(T)$ is defined in \eqref{eq:channel_ratio}.
\end{theorem}
\begin{proof}[Proof (Sketch)] Using the inequality \eqref{eq:KL_concentration} with $P_{Y^n} = P^{(\cC)}_{Y^n}$ and $t = c(T)\sqrt{\frac{n}{2}\log \frac{1}{1-2\eps}}$, we get
\begin{align*}
P_{Y^n|X^n=x^n}\Bigg[\log \frac{P_{Y^n|X^n=x^n}(Y^n)}{P^{(\cC)}_{Y^n}(Y^n)}  \ge D\Big(P_{Y^n|X^n=x^n} \Big\| P^{(\cC)}_{Y^n}\Big) 
+ c(T)\sqrt{\frac{n}{2}\log \frac{1}{1-2\eps}} \Bigg] \le 1-2\eps
\end{align*}
Now, just like Polyanskiy and Verd\'u, we can appeal to a strong converse result due to Augustin \cite{Augustin} to get
\begin{align}
\log M \le \log \frac{1}{\eps} + D\Big(P_{Y^n|X^n} \Big\| P^{(\cC)}_{Y^n} \Big|P^{(\cC)}_{X^n}\Big) 
+ c(T)\sqrt{\frac{n}{2}\log \frac{1}{1-2\eps}}. \label{eq:from_Augustin}
\end{align}
Therefore,
\begin{align*}
D\Big( P^{(\cC)}_{Y^n} \Big\| P^*_{Y^n}\Big) 
& = {D\Big(P_{Y^n|X^n} \Big\| P^*_{Y^n} \Big|P^{(\cC)}_{X^n}\Big)} - {D\Big(P_{Y^n|X^n} \Big\| P^{(\cC)}_{Y^n} \Big|P^{(\cC)}_{X^n}\Big)} \\
& \le {nC} - {\log M + \log \frac{1}{\eps} + c(T)\sqrt{\frac{n}{2}\log \frac{1}{1-2\eps}}},
\end{align*}
where the first step is by the chain rule, the second follows from the properties of the capacity-achieving output distribution, and the last step uses \eqref{eq:from_Augustin}.
\end{proof}
A useful consequence of this result is that a broad class of functions evaluated on the output of a good code concentrate sharply around their expectations with respect to the capacity-achieving output distribution:
\begin{theorem}
Consider a DMC $T$ with $c(T) < \infty$. Let $d$ be a metric on $\sY^n$, and suppose that $P_{Y^n|X^n=x^n}$, $x^n \in \sX^n$, as well as $P^*_{Y^n}$, satisfy $T_1(c)$ for some $c > 0$. Then, for every $\eps \in (0,1/2)$, every $(n,M,\eps)$-code $\cC$ for $T$, and every function $f \colon \sY^n \to \R$ which is $L$-Lipschitz on $(\sY^n,d)$, we have
\begin{align}
P^{(\cC)}_{Y^n}\Big( \left|f(Y^n) - \E[f(Y^{*n})]\right| \ge t\Big) 
\le \frac{4}{\eps} \; \exp\left( nC - \ln M + a\sqrt{n} - \frac{t^2}{8cL^2} \right),
\; \forall \, r \ge 0 \label{eq:good_code_con}
\end{align}
where $Y^{*n} \sim P^*_{Y^n}$, and $a \deq c(T)\sqrt{\frac{1}{2}\ln \frac{1}{1-2\eps}}$.
\end{theorem}
As pointed out in \cite{Polyanskiy_Verdu_good_codes},  concentration inequalities like \eqref{eq:good_code_con}  can be very useful for gaining insight into the performance characteristics of good channel codes without having to explicitly construct such codes: all one needs to do is to find the capacity-achieving output distribution $P^*_Y$ and evaluate $\E[f(Y^{*n})]$ for an arbitrary $f$ of interest. Consequently, the above theorem guarantees that $f(Y^n)$ concentrates tightly around $\E[f(Y^{*n})]$, which is relatively easy to compute since $P^*_{Y^n}$ is a product measure.

\bibliography{references}
\end{document}